\documentclass[11pt]{article}
\usepackage{fullpage}

\usepackage{graphics}
\usepackage[dvips]{epsfig}

\usepackage{amsmath}
\usepackage{amssymb}
\usepackage{amsfonts}
\usepackage{graphicx}

\newtheorem{theorem}{Theorem}[section]
\newtheorem{lemma}{Lemma}[section]

\newtheorem{proposition}{Proposition}[section]

\newcommand{\qed}{\hfill $\Box$ \bigbreak}
\newcommand{\alg}{\hfill $\diamondsuit$}
\newenvironment{proof}{\noindent {\bf Proof.}}{\qed}

\newcommand{\remove}[1]{}

\begin{document}

\baselineskip  0.2in 
\parskip     0.1in 
\parindent   0.0in 

\title{{\bf Use of Information, Memory and Randomization in Asynchronous Gathering }}

\author{
Andrzej Pelc\thanks{D\'{e}partement d'informatique, Universit\'{e} du Qu\'{e}bec en Outaouais,
Gatineau, Qu\'{e}bec J8X 3X7,
Canada. E-mail: pelc@uqo.ca.
}
}

\date{ }
\maketitle

\begin{abstract}
We investigate initial information, unbounded memory and randomization in gathering mobile agents on a grid. We construct a state machine, such that it is possible to gather, with probability 1, all configurations of its copies. This machine has initial input, unbounded memory, and is randomized. We show that no machine having any two of these capabilities but not the third, can be used to gather, with high probability, all configurations.

We construct deterministic Turing Machines that are used to gather all connected configurations, and we construct deterministic finite automata that are used to gather all contractible connected configurations.

\vspace{2ex}

\end{abstract}

\vfill

\vfill

\thispagestyle{empty}
\setcounter{page}{0}
\pagebreak

\section{Introduction}

Initial information, unbounded memory and randomization are features known to influence the feasibility of many distributed tasks. In particular, each of these capabilities
significantly influences the feasibility of tasks performed by mobile agents in networks. For example, even synchronous anonymous agents cannot meet
in an oriented two-dimensional grid or in ring, in a deterministic way, as they are incapable of breaking symmetry. On the other hand, such meeting of synchronous randomized agents can 
be accomplished with probability 1 on these networks, which shows the power of randomization in this case. The power of unbounded memory has been also frequently observed.
For example, it was shown in \cite{CKP} that, even in the class of rings, memory of logarithmic size is needed to accomplish rendezvous of deterministic  anonymous agents
starting from non-symmetric positions. This shows the power of unbounded memory: deterministic finite automata accomplish less, in terms of rendezvous, than deterministic unbounded-memory state machines. Finally, the impact of initial information on the feasibility of tasks performed by mobile agents is also well known.
For example, it is possible to perform exploration with stop of all oriented rings by a single agent, if some bound on the size of the ring is initially given to the agent, but it is impossible to perform this task in this class without any initial input.

In this paper we investigate the power of these three capabilities
in accomplishing a well-researched task: that of gathering mobile agents. Agents are modeled as identical (anonymous) copies of the same state machine, either with bounded memory (and hence they are
finite automata), or with unbounded memory (and hence they are computationally equivalent to Turing Machines). They can be deterministic or randomized.
They can get some initial input depending on the initial configuration of the agents, or they can be deprived of any such input. 

Agents start from distinct cells, they move asynchronously in the grid, and all of them have to get to a single cell and detect this fact. They have a compass
showing  the cardinal directions, and in each round they see the states of all agents in the same and in neighboring cells. On the basis of this observation in round $r$,
of the initial input, if it is provided, and of a coin toss in the case of randomized agents, an agent transits to some state and, depending on it, either stays inert or moves to a neighboring cell in some round $r'>r$ decided by the asynchronous adversary. 
The actions of an agent are performed in an increasing sequence $(r_1,r_1',\dots, r_k,r_k')$ of rounds, where in rounds $r_i$ the agent looks at states in its neighborhood, and in rounds $r_i'$ it transits to a new state and, depending on it, stays still or moves.


\subsection{Our results}
\label{subsec:ourresults}

We ask the fundamental question if there exists a state machine, such that it is possible to gather, with probability 1, all initial configurations of agents that are copies of this machine.
The answer is yes, and our first result is the construction of such a machine. This is a machine in the strongest of our models: it has initial input, unbounded memory,
and it is randomized. Hence it is natural to ask if those three capabilities are indeed necessary to gather all configurations. We answer this question affirmatively  as well,
by showing that no machine having any two of these capabilities but not the third, can be used to gather, with high probability, all possible configurations.

Finally, we consider the gathering power of deterministic machines. For this model it is natural to consider connected configurations, as deterministic machines
(even with initial input and with unbounded memory) cannot be used to gather very simple disconnected configurations.
 We construct deterministic Turing Machines (without initial input) that can be used to gather all connected configurations, and we construct deterministic finite automata 
 (without initial input) that can be used to gather all contractible connected configurations, i.e., connected configurations without ``holes''.


\subsection{Related work}
\label{subsec:relatwork}

Gathering has been mostly studied for two mobile agents and in this case it is usually called rendezvous.
An extensive survey of  randomized rendezvous in various scenarios  can be found in
\cite{alpern02b}, cf. also  \cite{alpern95a,alpern02a,anderson90,baston98,israeli}. 
Deterministic rendezvous in networks has been surveyed in \cite{Pe}.
Several authors
considered the geometric scenario (rendezvous in an interval of the real line, see, e.g.,  \cite{baston98,baston01,gal99},
or in the plane, see, e.g., \cite{anderson98a,anderson98b}).
Gathering more than two agents has been studied, e.g., 
in \cite{GWB,israeli,lim96,thomas92}. In~\cite{YY} the authors considered 
rendezvous of many agents with unique labels, and gathering many labeled agents in the presence of Byzantine agents was studied in \cite{DPP}. 
Configurations of anonymous agents that can be synchronously gathered in arbitrary graphs were characterized in \cite{DP}.
The problem was also studied in the context of multiple robot systems, cf.
\cite{CP05,fpsw}, and fault tolerant gathering of robots in the plane was studied, e.g., in \cite{AP06,CP08}. 

For the deterministic setting, a lot of effort has been dedicated to the study of the feasibility of rendezvous, and to the time required to achieve this task, when feasible. For instance, deterministic rendezvous with agents equipped with tokens used to mark nodes was considered, e.g., in~\cite{KKSS}. Deterministic rendezvous of two agents that cannot mark nodes but have unique labels was discussed in \cite{DFKP,TSZ07}.
These papers are concerned with the time of rendezvous in arbitrary
graphs. In \cite{DFKP} the authors show a rendezvous algorithm polynomial in the size of the graph, in the length of the shorter
label and in the delay between the starting time of the agents. In \cite{TSZ07} rendezvous time is polynomial in the first two of these parameters and independent of the delay.

Memory required by two anonymous agents to achieve deterministic rendezvous has been studied in \cite{FP2} for trees and in  \cite{CKP} for general graphs.
Memory needed for randomized rendezvous in the ring is discussed, e.g., in~\cite{KKPM08}. 

Several authors have investigated asynchronous rendezvous in the plane \cite{AGM,CFPS,fpsw} and in network environments
\cite{BCGIL,CLP,DGKKP}.
In the latter scenario it was assumed that the agent chooses the edge which it decides to traverse but the adversary controls the speed of the agent. Under this assumption rendezvous
in a node cannot be guaranteed even in very simple graphs, and hence the rendezvous requirement is relaxed to permit the agents to meet inside an edge. 
In \cite{BCGIL} the authors studied rendezvous in grids of two agents sharing a common coordinate system. The feasibility of asynchronous rendezvous of two anonymous agents in arbitrary graphs was discussed in \cite{GP}, both in the deterministic and in the randomized version.

Collective exploration of the grid in models similar to ours has been considered, e.g., in \cite{ELSUW,KLUW}.

\subsection{The model}

We consider the grid ${\bf Z} ^2$ whose nodes are called {\em cells}. The distance between two cells $(x,y)$ and $(x',y')$ is $|x-x'|+|y-y'|$.
Cells $(x,y)$ and $(x',y')$ are {\em neighbors}, if their distance is 1.
An agent situated in cell $(x,y)$ can move  to one of the neighboring cells $(x,y+1),(x+1,y),(x,y-1),(x-1,y)$ or stay put. The move to each of the above cells
is denoted $N,E,S,W$, respectively, and we use $P$ to denote the stay-put action. This move or the stay-put action depends on the current memory state of the agent, 
which in turn depends on the previous state, on what the agent sees in the cell in which it is situated and in the neighboring cells, and on a coin toss in the case of random machines. In the sequel we will use the terms ``move North'', ``the cell North of a given cell'', etc. The grid is oriented and the agents have compasses, so each agent knows the meaning of $N,E,S,W$. However, agents do not see the coordinates of the cells. In the beginning, agents occupy some cells, one agent per cell. The set of these cells (and, by abuse of terminology, of these agents) is called an {\em initial configuration}.

In all our models, agents are represented as identical copies of a state machine defined as follows. Consider a set $Q$ of states. This set may be finite or (countable) infinite.
We specify an element $q_0 \in Q$ as the initial state.
A deterministic machine (called a deterministic finite automaton if $Q$ is finite, and a deterministic infinite state machine if $Q$ is infinite) is specified by two functions.
The action function $\phi:Q \longrightarrow \{N,E,S,W,P\}$ indicates how the action of the agent depends on its current state. 
In order to describe the transition function $\pi$ which indicates how agents change states, we denote by $S$ the set of finite sets of couples from ${\bf Z} ^+\times Q$, where ${\bf Z} ^+$ is the set of positive integers. The transition function is a function $\pi:Q\times S\times S \times S \times S \times S \longrightarrow Q$ and is interpreted as follows. Consider a five-tuple $(a_P,a_N,a_E,a_S,a_W)$  of elements of S. Each of its terms is a set of the type $\{(n_1,q_1),\dots, (n_k,q_k)\}$. It is interpreted as the observation of $n_i$ agents with state $q_i$, for $i=1,\dots,k$. For the set $a_P$ this observation concerns the cell where the observing agent is situated, for the  set $a_N$ it concerns the cell North of it, and so on. Now if $\pi(q,a_P,a_N,a_E,a_S,a_W)=q'$, then the agent observing the respective numbers of agents in the respective states in the five cells to which it has currently access, transits from state $q$ to state $q'$. This captures the assumption that an agent can observe states of all agents in its own and in the neighboring cells, and acts depending on the result of this observation.

A random machine (called a random finite automaton if $Q$ is finite, and a random infinite state machine if $Q$ is infinite) is defined by a suitable modification of the transition
function (the action function remaining the same). By $b$ we denote a random bit (each value with probability 1/2). The transition function is now
$\pi:\{0,1\}\times Q\times S\times S \times S \times S \times S \longrightarrow Q$, and its value is determined as before but additionally depends on the value
of the random bit, which is the first term of the 7-tuple in the domain. Both in the deterministic and in the random case, each agent starts in the initial state $q_0$,
if no initial input is given.

Notice that, in the case of an infinite set $Q$ of states, the deterministic (resp. random) infinite state machine is computationally equivalent to the
deterministic (resp. random) Turing Machine. We will use the latter terminology, thus speaking of four models of agents: 
deterministic finite automata, deterministic Turing Machines, random finite automata, and random Turing Machines.

We need one more distinction. The above defined models correspond to machines without any initial input. However, we will also use machines with initial input 
(possibly depending on the initial configuration) given to all the agents at the beginning as a finite binary string (the same for all agents). Let $\cal I$ be the set of finite binary strings. Machines (of each of the above four types) {\em with initial input}
are defined similarly as before, using additionally an input function $f: {\cal I}  \longrightarrow Q$ interpreted as follows. If the machine is given initial input $I$, it starts in state $f(I)$ instead of state~$q_0$.

The last issue to specify is the way in which asynchrony of the agents is modeled. Since we want the agents to gather in some cell, rather than inside an edge,
we cannot use the adversary from \cite{BCGIL,CLP,DGKKP} that can decide the speed of an agent in each edge. Instead, we adopt the asynchronous adversary from
 \cite{ELSUW,KLUW}. Time is partitioned in rounds and each agent is activated in an increasing  sequence $(r_1,r_1',\dots, r_k,r_k')$ of rounds, decided by the adversary.  In a  round $r_i$, the agent 
 in current state $q$ looks at states of agents in its cell and in the neighboring cells  and computes the value of the function $\phi$ on this basis and possibly on the basis of the random bit. This is called a {\em look}. In the round $r_i'$, the agent transits to the new state $q'$ computed by $\phi$ and stays put or moves (i.e.,
 it executes the function $\pi$). As in  \cite{ELSUW,KLUW} we assume for simplicity that the change of state and the move are done at the end of round $r_i'$.   This means that in the segment $[r_i,r_i']$ of rounds the agent stays inert in state $q$. The agent can be seen in the (possibly new) cell
 in state $q'$ only in round $r_i'+1$.
 
 For any of the above models, the task of gathering is formulated as follows.
\begin{itemize}
\item
There exists a round in which all agents are in a single cell and all have transited to the final state $\omega$, which is interpreted as the agents detecting  gathering.
\item
Once in the state $\omega$, an agent remains in it and stays put forever.
\item
An agent never transits to the final state $\omega$ before all agents are in a single cell (no false detection of gathering on the part of any agent). 
\end{itemize}

 \section{Three qualities are better than two}
 
 We start by asking the following fundamental question.  
 
 \begin{quotation}
 \noindent
 Does there exist a state machine, such that it is possible to gather, with probability 1, all initial configurations of agents that are copies of this machine?
\end{quotation}

The following result gives a positive answer to this question.

 \begin{theorem}\label{yes}
There exists a random Turing Machine $M$ with initial input that has the following property. Consider any initial configuration $C$ in which agents are copies of $M$, and suppose that  the size $n$ of $C$ is given to each of them as initial input. Then the configuration will be gathered with probability 1. 
 \end{theorem}
 
 We present an algorithm that achieves gathering with probability 1, starting from any initial configuration, and can be executed by agents that are identical copies of some random Turing Machine with initial input.
 
 {\bf Algorithm} {\tt General Gathering} 
 
 The high-level  idea of the algorithm is the following. The initial input is the number $n$ of agents in the configuration. Each agent registers it in its memory. 
 We first design a procedure that brings agents, with probability 1, at distance at most 1, makes them notice this fact, and breaks symmetry between them.
 Call this procedure {\tt Basic Approach}. Each agent starts by executing it, all the time remembering the path to its starting cell, coded as a sequence of moves
 $N,E,S,W$, called the {\em characteristic} of the agent. The characteristic is extended by one term after every move.  Upon every approach  at distance at most 1 during the execution of {\tt Basic Approach},
 one of the agents, the {\em loser}, terminates the execution of this procedure and stays idle until further notice, and the other,  the {\em winner}, continues the procedure, remembering the path to the inert agent. 
 At all times, an agent that still continues procedure {\tt Basic Approach}, has a {\em bag} consisting of some agents. In the beginning, the bag of each agent is empty.
 When two agents get at distance at most 1 during {\tt Basic Approach}, the winner adds to its bag the bag of the losing agent and this agent itself. Hence, at all times, a bag consists of (temporarily) inert agents.  Since agents are anonymous, the way to remember them (i.e., to hold them in the bag) is by keeping for each agent in the bag a path to its waiting cell: this path is
 coded by moves $N,E,S,W$ and updated at each move. After each approach of agents at distance at most 1 during procedure  {\tt Basic Approach}, the number of agents still executing this procedure decreases.
 After the end of the execution of this procedure by the agent $a$ that always won, the bag of agent $a$ consists of $n-1$ agents. Agent $a$ notices it, and knows a path to each of the other agents. 
 Then $a$ terminates procedure {\tt Basic Approach} and  starts the procedure {\tt Guiding}. It consists in visiting all the other agents in some canonical order, and giving each of them a path to the starting cell of $a$.
 Whenever a waiting agent gets this path, it execute procedure {\tt Final}: it walks along the path given to it by $a$, and stops forever in the starting cell of $a$. After visiting all agents in its bag, agent $a$ executes procedure {\tt Final} in its turn: it gets back to its starting cell, using its own characteristic.
 At the end of procedure {\tt Final}, each agent that went back to the starting cell of $a$, transits to the final state $\omega$ upon noticing $n-1$ other agents in this cell.
 
 We now describe the details of the algorithm. We will use the following notions, introduced in  \cite{CLP}.
 (In \cite{CLP} they were defined for arbitrary graphs, we adopt them to the oriented grid).
 A {\em route} of an agent is the sequence of edges $(e_1,e_2,\dots)$, such that $e_i$ is incident to $e_{i+1}$, in the order in which the agent traverses them.
This sequence is a (not necessarily simple) path in the grid. A route can be coded as a sequence of letters $N,E,S,W$, the $i$th letter indicating which direction should be taken to traverse edge $e_i$.

Consider two routes $R_1$ and $R_2$ starting at nodes $v$ and $w$, respectively.
We say that these routes form a {\em tunnel}, if there exists a prefix $[e_1,e_2,\dots,e_n]$ of route $R_1$ 
and a prefix $[e_n, e_{n-1},\dots, e_1]$
of route $R_2$, for some edges $e_i$ in the grid, such that $e_i=\{v_i,v_{i+1}\}$, where $v_1=v$ and $v_{n+1}=w$. 
Intuitively, the route $R_1$ has a prefix $P$ ending at $w$ and the
route $R_2$ has a prefix which is the reverse of  $P$, ending at $v$.
For simplicity we will also say that prefixes $[e_1,e_2,\dots,e_n]$ and 
$[e_n, e_{n-1},\dots, e_1]$ form a tunnel. 

In \cite{GP} the authors showed a (randomized) algorithm that constructs a tunnel with probability 1,  for any starting nodes $v$ and $w$.
This was used to establish rendezvous with probability 1, of any two agents, under an asynchronous adversary different from ours. This adversary could walk an agent inside any edge chosen by the agent, with arbitrary, possibly varying speed, as long as the walk was continuous and ended up at the other endpoint of the edge. While such an adversary is at least as powerful as ours, it is impossible to guarantee meeting at a node, and hence in \cite{GP} (similarly as in \cite{CLP}) the meeting could occur inside an edge of the graph.
We do not allow such a possibility. Moreover, agents have to notice the approach in our case, while rendezvous (possibly inside an edge) in  \cite{GP} finished the
entire process. 
Hence we cannot use the rendezvous algorithm from \cite{GP},
and we need a deviation from it:  the tunnel constructed there has to be used (slightly) differently. 
This modified method of meeting between any two agents will be used as a building block in our gathering algorithm.

First notice that all routes in the grid can be ordered lexicographically, by (arbitrarily) ordering the cardinal directions $N<E<S<W$. The procedure
{\tt Basic Approach} can be described as follows. 

{\bf Procedure} {\tt Basic Approach}

The agent walks along the route produced in the algorithm from \cite{GP}. This route contains a random ingredient: in some rounds the agent tosses a coin and 
adds the obtained random bit to its partial label, on the basis of which a part of the route is constructed. Then it tosses a coin again, produces an extension of its partial label,  constructs another part of the route, and so on.

We will show that, regardless of the actions of the adversary, if there are still at least two agents  executing procedure {\tt Basic Approach},  there exists a round in which one of them
looks at states of agents in its cell and in the neighborhood  and sees one or more other agents  executing procedure {\tt Basic Approach} in one or more of these cells. 
The agent that made this observation enters a {\em contest}. It compares its characteristic to those of these other agents. If some other agent has a lexicographically larger characteristic, the observing agent 
becomes a {\em loser}, terminates its execution of
{\tt Basic Approach},  and stays idle  until it is visited by an agent performing procedure {\tt Guiding} (see below).  The agent whose characteristic is lexicographically largest of all agents executing procedure {\tt Basic Approach} that it currently sees becomes a {\em winner}. It waits until all the agents that lost with it realize that they lost and become losers,  and
it adds to its bag (the paths of) all
agents from the bags of all the losers and the paths (of length 0 or 1) to the losers themselves.
If the observing agent sees another agent in a neighboring cell with the same characteristic (which can happen e.g., at the beginning, when agents are in their starting neighboring cells, and their characteristics are empty sequences), the agent South or West of the other agent loses the contest.
The winning agent continues executing the procedure  {\tt Basic Approach} until either it becomes a loser, or until the number of agents in its bag is $n-1$. We say that in the first case it ends procedure {\tt Basic Approach} as a loser, and in the second case it ends it as a {\em champion}. These are values of variable $role$.
\alg

\pagebreak

{\bf Procedure} {\tt Guiding}

This procedure is executed only by the agent $a$ that finished procedure  {\tt Basic Approach} as the champion in some cell $v$.
The agent has paths from $v$ to the waiting cells of $n-1$ agents in its bag. It visits them in lexicographic order of the paths, each time coming back to $v$.
At each visit, the agent $a$ gives the visited agent a path to the starting cell of $a$. The procedure is terminated after coming back after the last visit.
\alg


{\bf Procedure} {\tt Final}

If the executing agent terminated procedure  {\tt Basic Approach} as a loser then upon getting the path to the starting cell of the champion, it goes to this cell along this path. If the executing agent terminated procedure  {\tt Basic Approach} as a champion, then, after terminating procedure {\tt Guiding}, it goes to its starting cell, along its characteristic.
In both cases the agent stays in this cell forever, and transits to the final state $\omega$ upon noticing $n-1$ other agents in this cell.
\alg

Now Algorithm {\tt General Gathering}  can be succinctly formulated as follows.

\begin{center}
\fbox{
\begin{minipage}{7cm}

{\bf Algorithm} {\tt General Gathering}

\vspace*{0.5cm}

execute procedure {\tt Basic Approach}\\
{\bf if} $role=champion$ {\bf then}\\
\hspace*{1cm}execute procedure {\tt Guiding}\\
execute procedure {\tt Final}

\end{minipage}
}
\end{center}

{\bf Proof of Theorem \ref{yes}}

The theorem is proved by showing that Algorithm {\tt General Gathering}  gathers an arbitrary initial configuration with probability 1.
In the rest of the proof, we omit the phrase ``with probability 1'', but all assertions should be understood with this qualifier.
In \cite{GP} the authors showed that the routes that they constructed, and that are used in our procedure {\tt Basic Approach},
 form a tunnel,  for any starting nodes $v$ and $w$.
 
 We first show that, for any two agents executing procedure {\tt Basic Approach}, either one of them becomes a loser as a result of approach with another agent, or
 they get at distance at most 1, notice this fact, and break symmetry between them. Consider any such agents $a$ and $b$,  starting at cells $v$ and $w$,
 respectively.
 If they both continue executing the procedure, since their routes form a tunnel, at some point one of the following two events must happen:
 either the agents traverse the same edge in opposite directions in the same round, or there exists a cell at which they are both in some round,
 coming from different directions. 
 
 In the first case, suppose that the crossing occurred in some round $\rho$. Then, in some round $\rho'>\rho$ one of the agents looks at the states in its neighborhood
 and realizes the existence of the other agent. It suspends its walk and waits until the other agent proceeds with the look in some round $\rho''\geq \rho'$.
 Then the two agents compare the paths from their starting cells, which must be different because the last move was from different directions. This breaks symmetry:
 one of the agents becomes the loser in this contest, and the other the winner. 
 
 In the second case, consider the earliest round in which both agents are at the same cell $v$, coming from different directions. If they come to this cell in the same round, the argument is as before. Hence assume that agent $a$ comes to cell $v$ from cell $u$ in round $r$, and agent $b$ comes to this cell from cell $u'$ in round $r'>r$.
 Let $\ell'$ be the round in which agent $b$ makes its last look before its move in round $r'$. Consider two cases. If $\ell'<r$, then agent $b$ is at cell $u'$ in round $r$ in which agent $a$ gets to $v$. Then it also gets to $v$ in round $r'$. At the first look of $a$ after its move to $v$ in round $r$, agent $a$ is at $v$ and agent $b$ is 
 either at $u'$ or at $v$. Then $a$ notices $b$, suspends its walk and waits until $b$ makes the next look. This can be either when $b$ is at $v$ or when it is at $u$.
 In both cases $b$ notices agent $a$, and the symmetry
 between the agents is broken as before: one of them becomes the loser and the other the winner. 
 
 If $r \leq \ell' <r'$, then in round $\ell'$ when agent $b$ makes a look, $b$ is at cell $u'$ and $a$ is at cell $v$. Hence $b$ notices agent $a$, suspends its walk,
 and waits until $a$ makes the next look. This can happen either when $a$ is still at $v$ or when $a$ is at $u'$. In both cases $a$ notices agent $b$ and the symmetry
 between the agents is broken as before: one of them becomes the loser and the other the winner.
 
 Hence, if there are still at least two agents executing procedure {\tt Basic Approach}, one of them must eventually enter a contest. Since every contest produces at least
 one loser, the number of agents executing procedure {\tt Basic Approach} decreases after each of them, until only one agent remains: the champion.
 At this point all agents are inert, and the champion knows paths to each of them. The champion $a$ visits each other agent, and gives it the current characteristic of $a$.
 Each such cell gets to the starting cell of the champion and remains there. The champion also gets to this cell, following its characteristic after visiting the last loser.
 Eventually each agent realizes that there are $n-1$ other agents in this cell and enters the final state $\omega$, which concludes the task of gathering. Since agents transit to the final state only when seeing $n-1$ other agents in their cell, false detection of gathering cannot occur.
 
 Finally, we note that the task is accomplished in the event that is the intersection of less than $n$ events each of which is the approach of two agents, and thus has probability 1, in view of \cite{GP}. It follows that the intersection of these events also has probability 1, which concludes the proof.
\qed

{\bf Remark.} It should be noted that our algorithm giving the result from Theorem \ref{yes} uses our general assumption that during each {\em look} the agent
learns states of all agents situated in the same and in neighboring cells. Can this assumption be weakened by allowing the possibility of  learning only states of agents located in the same cell? The answer is no. It is easy to observe that even in the simplest scenario of two agents walking on the infinite line graph (and hence simpler than the infinite grid we are considering)  meeting of the agents with probability 1 is impossible if they cannot ``see'' beyond the currently occupied cell. This implies that our assumption that agents can sense the immediate neighborhood is necessary. It also shows the significant difference between our scenario where agents must meet in a cell and that from \cite{GP}, where they could meet inside an edge.
While establishing a tunnel was enough to produce a meeting with probability 1 in \cite{GP}, without the possibility of  any sensing, {\em it would not  be enough} to produce a meeting
 with probability 1 in a cell, when agents are incapable of sensing beyond the currently occupied cell. This is why contests in the procedure  {\tt Basic Approach} (using sensing of the neighborhood) are needed in our current scenario. \qed

 Since the machine whose existence is asserted in Theorem \ref{yes} is of the strongest of all types that we consider, i.e., it has initial input, unbounded memory and is randomized, it is natural to ask if those three capabilities are indeed necessary to gather, with high probability, all initial configurations. We answer this question affirmatively as well.
 To do so, we consider the three possible combinations of two of these capabilities, without the third, and we show initial configurations that cannot be gathered,
 with high probability,
 if agents are copies of any such weaker machine.
 
 We start our negative results  with the following simple observation. Consider any deterministic Turing Machine with initial input, and consider an initial configuration
 consisting of two agents at distance $d$ larger than 1, that are copies of this machine. Consider any initial input $I$ given to these agents, and let $q=f(I)$ be the state
 in which agents start. The adversary starts the two agents in the same round $r_0$ and then activates each of them in every round. In every round $r \geq r_0$
 the agents will be in the same state, they will not see any other agent,  they will make the same moves, and hence they will be at the same distance $d$. Thus they can never meet. This proves the following proposition showing that randomness is necessary.
 
 \begin{proposition}\label{one}
 No deterministic Turing Machine (regardless of the initial input) can be used to gather all possible initial configurations.
 \end{proposition}
 
 We proceed to prove that some initial input is necessary as well, for gathering with high probability. This is the subject of the following result.
  
 \begin{theorem}\label{two}
 Consider any random Turing Machine $M$ without initial input. There exists an initial configuration that will be gathered with probability at most 1/2,
 if all agents are copies of $M$. 
 \end{theorem}
 
 \begin{proof}
 Suppose to the contrary that there exists a random Turing Machine $M$ without initial input, such that any initial configuration can be gathered with probability larger than 1/2, if all agents are copies of $M$.
 Consider the initial configuration $C$ consisting of two agents situated in cells $(0,0)$ and $(0,2)$. Consider the adversary $A$ that activates both agents simultaneously in every round until gathering.
 There exists a positive integer $t$ such that gathering of configuration $C$ under adversary $A$ will occur in time at most $t$ with probability larger than 1/2. Let $E$ be this event. Now consider the initial configuration $C'$ consisting of three
 agents situated in cells $(0,0)$,  $(0,2)$ and $(0, t+3)$. Consider the adversary $A'$ that activates the agents starting from cells $(0,0)$ and $(0,2)$ simultaneously in every round until gathering, and that does not activate the third agent in any of the first $t$ rounds. 
 Let $E'$ be the event that gathering of configuration $C'$ will eventually occur under adversary $A'$. Events $E$ and $E'$ are not disjoint because each of them has probability larger than 1/2.
 Let $e$ be an elementary event in $E \cap E' $. Under event $e$ and under adversary $A$, configuration $C$ is gathered in time at most $t$. Consider what happens with configuration $C'$ under event $e$ and under adversary $A'$. In the first $t$ rounds, agents starting in cells $(0,0)$ and $(0,2)$ behave exactly as in configuration $C$ under adversary $A$ and event $e$ because in each of these
 rounds they compute the same value of  function $\psi$: the third agent is inert in these rounds and each of the two agents is too far from it to reach it in these rounds.
 Hence agents starting in cells $(0,0)$ and $(0,2)$ in configuration $C'$ meet and enter the final state $\omega$ in time at most $t$, which is incorrect, since they have not met the third agent.
 This contradiction concludes the proof.
 \end{proof}
 
 Finally, we want to show that unbounded memory is also necessary for gathering with high probability.  Indeed, we have the following result.
 
  \begin{theorem}\label{three}
 Consider any random finite automaton $M$ with initial input. 
 Consider any initial input $I_C$ given to agents in an initial configuration $C$, where all agents are copies of $M$. 
 Then there exists a configuration
 $C$ that will be gathered with probability at most 1/2. 
 \end{theorem}
 
 \begin{proof}
 Suppose to the contrary that there exists a random finite automaton $M$ with initial input, such that any initial configuration can be gathered with probability larger than 1/2, if all agents are copies of $M$.
 We construct by induction the following sequence of initial configurations.
 The initial configuration $C_1$ consists of two agents situated in cells $(0,0)$ and $(0,2)$.  
 Suppose by induction that the initial configurations $C_i$, for $i<j$, have been already constructed.
 Let $I_i$ be the initial input given to the agents when $C_i$ is the initial configuration.
 Let $q_i=f(I_i)$ be the state in which the agents start, given input $I_i$. 
 Consider the adversary $A_i$ that activates all agents in $C_i$ simultaneously in every round until gathering.
 There exists a positive integer $t_i$ such that gathering of configuration $C_i$ under adversary $A_i$ will occur in time at most $t_i$ with probability larger than 1/2.
 The configuration $C_j$ consists of all agents in $C_{j-1}$ plus  an additional agent situated at distance  at least $t_1+t_2+\dots +t_{j-1}+1$ from all agents in $C_{j-1}$.
 By definition, initial configurations $C_i$ are an increasing sequence of sets of agents ordered by inclusion.
 
 Let $k$ be the smallest integer such that there exists an integer $k'<k$ for which $q_k=q_{k'}$. Such an integer must exist because the number of states in machine $M$ is finite.
 
  Let $E$ be the event that gathering of configuration $C_{k'}$ under adversary $A_{k'}$ occurs in time at most $t_{k'}$ with probability larger than 1/2. 
Consider the following adversary $A'$ for configuration $C_k$: it activates the agents from  the sub-configuration $C_{k'}$ simultaneously in every round until gathering, and it does not activate any other agent from $C_k$  in any of the first $t_{k'}$ rounds. 
 Let $E'$ be the event that gathering of configuration $C_k$ will eventually occur under adversary $A'$. Events $E$ and $E'$ are not disjoint because each of them has probability larger than 1/2.
 Let $e$ be an elementary event in $E \cap E' $. Under event $e$ and under adversary $A_{k'}$, configuration $C_{k'}$ is gathered in time at most $t_{k'}$. 
 
 Consider what happens with configuration $C_k$ under event $e$ and under adversary $A'$. 
 In the first $t_{k'}$ rounds, agents from the sub-configuration $C_{k'}$ behave exactly as under adversary $A_{k'}$ and under event $e$ because in each of these
 rounds they compute the same value of  function $\psi$: the other agents from configuration $C_k$ are inert in these rounds and each of these agents is too far from agents of the sub-configuration $C_{k'}$ to be reached by them in these rounds.
 Hence agents from the sub-configuration $C_{k'}$ meet and enter the final state $\omega$ in time at most $t_{k'}$, which is incorrect, since they have not met the other agents of $C_k$.
 This contradiction concludes the proof.
\end{proof}
 
Proposition \ref{one} and Theorems \ref{two} and \ref{three} show that all three considered features, i.e.,  initial input, unbounded memory and randomness, must be used to construct a state machine that will permit to gather all possible initial configurations with high probability. Since Theorem \ref{yes} asserts the existence of such a machine that can be used for gathering all initial configurations with probability 1, this result is the best possible in terms of machine strength assumptions.

Finally, notice that, while our positive result holds even for the omniscient asynchronous adversary that knows all random bits in advance, our negative results hold even for the oblivious adversary that does not such advance knowledge.

 \section{Deterministic machines}
 
 In this section we consider the power of deterministic machines (both Turing Machines and finite automata) to gather the important class of connected configurations.
  These are initial configurations for which the subgraph of the grid induced by cells containing agents is connected.
 In the case of deterministic machines it is a natural class to consider, as even some of the simplest disconnected configurations (two agents at distance larger than 1) cannot be gathered deterministically (regardless of the initial input) because (as we have observed in Section 2) symmetry cannot be broken, if the adversary activates the agents simultaneously.
 
 We show that deterministic Turing Machines without initial input can be used to gather all connected configurations, and that deterministic finite automata
 without initial input (our weakest model) can be used to gather all contractible connected configurations, i.e., connected configurations without ``holes''.  The latter configurations are formally defined as follows:
 both the subgraph of ${\bf Z} ^2$ induced by cells containing agents and the subgraph of ${\bf Z} ^2$ induced by cells not containing agents, are connected.
 
 The first result of this section concerns the class of all connected configurations.
 
  \begin{theorem}\label{detTM}
  There exists a deterministic Turing Machine $M$ without initial input, such that any initial connected configuration will be gathered, if all agents are copies of $M$.
  \end{theorem}
  
  We present an algorithm that achieves gathering, starting from any initial connected configuration, and can be executed by agents that are identical copies of some Turing Machine without input.
  
  {\bf Algorithm} {\tt Connected Gathering} 
 
 The high-level  idea of the algorithm is the following. First, all agents acquire a map of the configuration, executing procedure {\tt Map Construction}. This is done without moving any agent, solely by writing in the memory of the agents 
 increasingly longer sequences of letters $N,E,S,W$, and of symbols $N',E',S',W'$, corresponding to paths in an (asynchronous) DFS search of the graph representing the configuration, starting from every agent. Once the map is acquired by an agent and it becomes aware of it,
 the agent gets the tag {\em informed} and starts participating in procedure {\tt Confirmation} whose aim is to learn that all agents are {\em informed}. Procedure {\tt Confirmation} is essentially
 a repetition of procedure {\tt Map Construction}, but only {\em informed} agents participate in it.  Finally, after completing procedure {\tt Confirmation}, 
 each agent executes procedure {\tt Walk}. It registers the
 total number $n$ of agents, and identifies the East-most agent $a$ in the North-most row of agents. Then each other agent individually goes to the cell occupied by $a$. Every agent transits to the final state $\omega$ when it sees $n-1$ other agents in this cell. 
 
 We now describe the details of the algorithm. 
 A {\em clean} sequence is a sequence of letters from $\{N,E,S,W\}$. Concatenation of sequences is denoted by $\cdot$. 
The {\em mirror} of the sequence $\alpha=(a_1,\dots,a_k)$, where $a_i\in \{N,E,S,W\}$, is the sequence
$(a_k',a_{k-1}',\dots,a_1')$ denoted as $\alpha'$. 
 A {\em leaf} is an agent that has only one neighboring agent.

 The first procedure can be described as follows.
 
 {\bf Procedure} {\tt Map Construction} 
 
 Part 1.
 
 1. After the first look, an agent $a$ that is not a leaf and sees another agent $b$ in a neighboring cell, writes in its own memory the sequence consisting of a single letter $x$ corresponding to the direction in which $a$ is with respect to $b$. This codes a forward arrow in the DFS search.
 
 2. After the first look, an agent $a$ that is a leaf and sees another agent $b$ in a neighboring cell, writes in its own memory the sequence $(xx')$, where $x$ is the letter corresponding to the direction in which $a$ is with respect to $b$. This codes a forward and the opposite backward arrow in the DFS search, when a leaf is visited.

 3. If in some look an agent $a$ that is not a leaf sees in the memory of a neighbor agent $b$ a clean sequence $\alpha$ that was not there at the previous look of $a$, agent $a$ appends to its own memory the sequence $\alpha\cdot(x)$, where $x$ corresponds to the direction in which $a$ is with respect to $b$. This codes a forward arrow in the DFS search. 
 
 4.  If in some look an agent $a$ that is a leaf sees in the memory of a neighbor agent $b$ a clean sequence $\alpha$ that was not there at the previous look of $a$, agent $a$ appends to its own memory the sequence $\alpha\cdot(xx')$, where $x$ corresponds to the direction in which $a$ is with respect to $b$. This codes a forward arrow and the opposite backward arrow in the DFS search, when a leaf is visited.
  
  5.  If in some look an agent $a$ sees in the memory of a neighbor agent $b$ a clean sequence of the form $\alpha \cdot \beta$ that was not there at the previous look of $a$, and such that $\alpha$ is in the memory of $a$, then $a$ appends to its own memory the sequence $\alpha \cdot \beta \cdot (xx')$, where $x$ is the letter corresponding to the direction in which $a$ is with respect to $b$. This codes a forward and the opposite backward arrow in the DFS search, when a loop in the search is closed.

6.  If in some look an agent $a$ sees in the memory of a neighbor agent $b$ a sequence $\alpha$, and it sees, in the memories of all other neighbor agents, sequences of the form
 $\alpha \cdot (x) \cdot \beta \cdot \beta '$, where $x\in \{N,E,S,W\}$, and it has in its own memory the sequence  $\alpha \cdot (x)$, agent $a$ appends to its memory the extension $\alpha \cdot (x) \cdot \beta \cdot \beta '  \cdot (x')$ of each sequence $\alpha \cdot (x) \cdot \beta \cdot \beta ' $
 and erases the sequence $\alpha \cdot (x)$ from its memory. This codes a backward arrow in the DFS search. 
 
 Part 2.
 
 An agent $a$ that in some look sees only sequences of the form $\alpha \cdot \alpha'$ in the memories of all
 neighbor agents, gets the tag {\em informed} and constructs the map of the configuration as follows. For every sequence $\alpha \cdot \alpha'$  in the memory of every neighbor of $a$, and for every prefix $\beta$ of $\alpha$ in every such sequence, agent $a$ puts (in its memory) an agent in the cell corresponding to the path $\beta$ from its own cell. Also, agent $a$ records the total number $n$ of agents.
 \alg
 
 The second procedure, whose aim is learning that all agents are {\em informed}, is a repetition of the previous one among agents with tag {\em informed}, and can be described as follows.
 
 {\bf Procedure} {\tt  Confirmation} 
 
 The first part of the procedure is a copy of Part 1 of procedure {\tt Map Construction}, with the following modification: each word ``agent'' is replaced by ``agent with tag {\em informed}''.
 All sequences inscribed in the memories of agents during procedure {\tt  Confirmation} are written in a different color (say red) than in procedure {\tt Map Construction}.
 
 The second part of the procedure is as follows.
 An {\em informed} agent $a$ that in some look sees the tag {\em informed} in all neighboring agents, and all red sequences that it sees in their memories are of the form $\alpha \cdot \alpha'$,
 gets a tag {\em ready}. At this point the agent knows that all other agents know the map of the configuration.
 \alg
 
 The aim of the final procedure is getting to the same cell. It is executed by each agent with tag {\em ready}.

 {\bf Procedure} {\tt  Walk} 
 
 Let $(x,y)$ be the cell of the agent, let $y+y'$ be the North-most row of the configuration, and let $(x+x', y+y')$, for some integers $x',y'$, be the East-most cell occupied by an agent in this row.
 (Note that $y'$ is necessarily non-negative, but $x'$ may be negative.) The agent makes $y'$ steps North. Then it makes  $x'$ steps East, if $x'$ is non-negative, and $x'$ steps West, if 
 $x'$ is negative. At this point the agent stops forever. The agent transits to the final state $\omega$ when it sees $n-1$ other agents in this cell.
 \alg
 
 Now Algorithm {\tt Connected Gathering}  can be succinctly formulated as follows.

\begin{center}
\fbox{
\begin{minipage}{7cm}

{\bf Algorithm} {\tt Connected Gathering}

\vspace*{0.5cm}

execute procedure {\tt Map Construction} \\
execute procedure {\tt  Confirmation} \\
execute procedure {\tt Walk}

\end{minipage}
}
\end{center}

{\bf Proof of Theorem \ref{detTM}}

The theorem is proved by showing that Algorithm {\tt Connected Gathering} (which is deterministic and does not use any initial input) correctly accomplishes gathering,
and all agents detect this fact.

Procedure {\tt Map Construction} is an implementation of an (asynchronous) Depth First Search of the graph representing the connected configuration,
starting from every cell of this graph. Clean sequences code forward paths in DFS, and sequences with primed letters code backward paths, and are needed to signal when a given DFS call is completed (item 6 of Part 1 of procedure  {\tt Map Construction}). In particular, any agent $a$ all of whose neighbors completed
the search initiated by $a$ learns it and becomes {\em informed} (Part 2 of procedure  {\tt Map Construction}). Keeping all forward paths followed by corresponding backward paths in the memory of the agents serves to reconstruct the configuration, which is done in Part 2 of procedure  {\tt Map Construction}. Hence, at the end of procedure  {\tt Map Construction}, an agent has a correct map of the entire configuration and gets the tag {\em informed}. 

At this point, every agent can unambiguously identify the East-most cell of the North-most row of the configuration. However, trying to get there immediately, before making sure that all other agents are  {\em informed}, would be dangerous, as it would risk to disconnect the configuration, preventing other (slower) agents to see the
East-most cell of the North-most row of the (initial) configuration. Preventing this is the role of procedure  {\tt  Confirmation}. Since this procedure is a DFS executed only by  {\em informed} agents, when it is completed, every agent can safely get to the designated cell, knowing that all other agents are {\em informed}, and hence will get independently to the same cell.
The moves of each agent are done only during procedure {\tt Walk}, and end when the agent gets to this cell.  Since agents transit to the final state only when seeing $n-1$ other agents in their cell, false detection of gathering cannot occur.
\qed

 Our final result concerns the class of contractible connected configurations. Recall that these are configurations for which both the subgraph of ${\bf Z} ^2$ induced by cells containing agents and the subgraph of ${\bf Z} ^2$ induced by cells not containing agents, are connected (the latter subgraph is infinite). Notice that these are exactly connected configurations without holes, where a {\em hole} is defined as a finite connected component of the subgraph of ${\bf Z} ^2$ induced by cells not containing agents.

  \begin{theorem}\label{fa}
  There exists a deterministic finite automaton $M$ without initial input,  such that any initial connected contractible  configuration will be gathered, if all agents are copies of $M$.
  \end{theorem}
  
  We present an algorithm that achieves gathering, starting from any initial connected contractible configuration, and can be executed by agents that are identical copies of some deterministic finite automaton without initial input. We will need the following terminology. A {\em N-leaf} (resp. {\em E-leaf}, {\em S-leaf}, {\em W-leaf}) is an agent all of whose neighbor agents are located in the cell South of it (resp. West of it, North of it, East of it). A {\em NE-corner} (resp.  {\em NW-corner})
  is an agent that has neighbor agents in exactly two cells, South of it and West of it (resp. South of it and East of it).
A {\em 4-cycle} is a set of agents forming a 4-cycle in the graph representing the configuration. An {\em articulation cell} in a connected configuration is a cell whose removal disconnects the configuration.
  
  {\bf Algorithm} {\tt Connected Contractible Gathering}
  
  The high-level idea of the algorithm is the following. Agents use only local observations to decide on their moves (as opposed to unbounded sequences of such previous observations, as in the case of
   Algorithm {\tt Connected Gathering}) : this is why the algorithm can be executed by copies of a finite automaton. Moreover, the algorithm has the following {\em shrinking property}: no agent ever moves to a cell that is not occupied by some agent in the current configuration. Thus the set of occupied cells is monotonically decreasing as the algorithm proceeds. Also, at all times, the current configuration is connected contractible. The algorithm proceeds using the following two types of moves: {\em leaf destruction} in which a leaf goes to its only neighboring occupied cell (with the aim of eventually emptying the cell containing this leaf), and {\em 4-cycle destruction}
  in which an agent that is a NE-corner or a NW-corner and is in a 4-cycle, goes to the cell South of it (with the aim of eventually destroying this 4-cycle, after all agents of this corner finally get South). It will be proved that, starting from any connected contractible configuration, any sequence of these two types of moves eventually results in all agents gathered in a single cell, at which point they all detect gathering and transit to the final state.
  
  We now describe the details of the algorithm. The first procedure removes leaves from the configuration. This is implemented by moving all leaf agents to the neighbor occupied cell. However, care should be taken in the special case of configurations where only two (neighbor) cells are occupied. Then agents in both these cells are  leaves, and careless moving to the neighbor cell could result in an infinite series of swaps. This is why, in this special case, symmetry between the two cells is broken, and agents of one of them stay still, while only the other agents move.

  {\bf Procedure} {\tt Leaf Destruction}
  
  If in some look an agent sees that it is a N-leaf or a E-leaf, it transits to the state {\em move South}, resp. {\em move West} and moves South (resp. moves West).
   If in some look an agent sees that it is a S-leaf or a W-leaf, it transits to the state {\em leaf asking}. An agent that sees a neighbor agent South or West of it in the state 
 {\em leaf asking} and is not a leaf itself, transits to the state {\em leaf agree}. An agent that is in state  {\em leaf asking} and sees all neighbor agents in state {\em leaf agree},  transits to the state {\em move North}, if it is an S-leaf, and to the state  {\em move East}, if it is a W-leaf, and moves North, resp. moves East.
 \alg
   
{\bf Remark.} Notice that, while the state {\em leaf agree} is introduced in the above procedure to prevent infinite swapping between two neighboring leaves, one such swap can still occur
(which is harmless). Due to asynchrony, the following situation could happen. There are three agents: agent $a$ in cell $(x,y)$, agent $b$ in cell
$(x,y+1)$, and agent $c$ in cell $(x,y+2)$. Agent $a$ that is a S-leaf, transits to the state {\em leaf asking}. Agent $b$ that is currently not a leaf, transits to the state
{\em leaf agree}. Agent $a$ makes a look, sees agent $b$ in state {\em leaf agree}, and computes its next state {\em move North}.
Then agent $c$ that is a N-leaf transits to the state {\em move South} and moves South. Now both $b$ and $c$ are N-leaves. They may look, realize it, transit to the state {\em move South} and move South,  while agent $a$ transits to the (previously computed) state {\em move North} and moves North, which results in a swap.
However, this cannot occur the second time because now there are only two occupied neighboring cells, and agent $a$ which is now a N-leaf will never transit to
state {\em leaf agree}. 

The next procedure removes NE-corners and NW-corners from 4-cycles in the configuration, by moving the respective agents South. Intuitively, an agent that is a NE-corner or a NW-corner, asks the agent
located South of it, if they are in a 4-cycle. After a positive answer, the agent moves South.

{\bf Procedure} {\tt 4-cycle Destruction}

 If in some look an agent sees that it is a NE-corner, then it transits to state {\em NE-question}. An agent that sees an agent North of it in state {\em NE-question} and sees an agent West of it,
 transits to state {\em NE-agree}. An agent that is in state {\em NE-question} and sees an agent South of it in state {\em NE-agree}, transits to state {\em move South} and moves South.
 
  If in some look an agent sees that it is a NW-corner, then it transits to state {\em NW-question}. An agent that sees an agent North of it in state {\em NW-question} and sees an agent East of it,
 transits to state {\em NW-agree}. An agent that is in state {\em NW-question} and sees an agent South of it in state {\em NW-agree}, transits to state {\em move South} and moves South.
 \alg
 
 Now Algorithm {\tt Connected Contractible Gathering}  can be succinctly formulated as follows.

\begin{center}
\fbox{
\begin{minipage}{12cm}

{\bf Algorithm} {\tt Connected Contractible Gathering}

\vspace*{0.5cm}

After each look {\bf do}

\hspace*{1cm}execute procedure {\tt Leaf Destruction} \\
\hspace*{1cm}execute procedure {\tt  4-cycle Destruction} \\
\hspace*{1cm}{\bf if} no agents are in the neighborhood {\bf then} transit to state $\omega$.

\end{minipage}
}
\end{center}

Te proof of correctness of Algorithm {\tt Connected Contractible Gathering} uses the following crucial geometric lemma.

\begin{lemma}\label{geo}
Every connected contractible configuration that has neither leaves nor 4-cycles containing NE-corners or NW-corners, must have all agents located in a single cell.
\end{lemma}

\begin{proof}
Suppose that there exists a connected contractible configuration $C$ that has neither leaves nor 4-cycles containing NE-corners or NW-corners, but has agents in at least two cells. 

Part 1.

We first prove the lemma under the additional assumption that the configuration does not have any articulation cell.
Let $x$ be the North-most row of the configuration, and let $(x,y)$ and $(x,y')$ be, respectively, the East-most and the West-most cells in this row.
We may assume that $y'<y$ because otherwise the cell $(x,y)$ would be a leaf, or all agents would be located in a single cell. By definition, there are no agents in cells $(x+1,y)$ and $(x,y+1)$,  hence this cell is a NE-corner.
There must be agents in cells $(x,y-1)$ and $(x-1,y)$ because otherwise the cell $(x,y)$ would be either a leaf, or this cell of the configuration would be isolated.
The is no agent in cell $(x-1,y-1)$ because otherwise there would be a 4-cycle containing a NE-corner.

There must exist in the configuration either a path from cell $(x,y')$ to cell $(x,y-1)$, or a path from cell $(x,y')$ to cell $(x-1,y)$ because otherwise the set of three cells
$(x,y)$, $(x+1,y)$ and $(x,y+1)$ would form a connected component, contradicting the fact that the configuration is connected. If all the paths from cell $(x,y')$ to cell $(x,y-1)$ contained cell $(x,y)$,
then cell $(x-1,y)$ would be an articulation cell, contrary to our assumption. Similarly, if all the paths from cell $(x,y')$ to cell $(x-1,y)$ contained cell $(x,y)$,
then cell $(x,y-1)$ would be an articulation cell. Hence, there must exist a path $\pi_1$ from cell $(x,y')$ to cell $(x,y-1)$ not containing cell $(x,y)$, and a path
$\pi_2$ from cell $(x,y')$ to cell $(x-1,y)$ not containing cell $(x,y)$. The union of these two paths disconnects the subgraph of ${\bf Z} ^2$ induced by cells not containing agents:
the cell $(x-1,y-1)$ is not in the infinite connected component of this subgraph. This contradiction completes the proof of the lemma under the additional assumption that the  configuration does not have any articulation cell.

Part 2.

Next suppose that the configuration $C$ contains articulation cells. Consider the smallest articulation component  $C'$ of the configuration $C$, i.e., a connected component of the configuration resulting
from removal of all articulation cells. Configuration $C'$ does not contain any articulation cells. Consider an articulation cell $v$ of $C$ and a cell $w$ of $C'$ adjacent to $v$.
We may suppose that $w$ is in the North-most row of $C'$, otherwise the argument is as in Part 1, with the cell $(x,y')$ replaced by $v$. 
We consider 4 cases: when $w$ is North, East, South or West of $v$. 

Case 1. $w$ is North of $v$.

The cell $w$ cannot be a single cell of $C'$ in this row because then $w$ would be a leaf of $C$. 
Hence either there is a cell East of $w$ or West of $w$ in $C'$. Consider the first case. Suppose that $w'\neq w$ is the East-most cell in this row.  If $w'$ is a neighbor of $w$ then it must have a South neighbor, otherwise it would be a leaf.
Hence $w'$ is a NE-corner in a 4-cycle, which is a contradiction. Hence $w'$ is not a neighbor of $w$, and the argument is as in Part 1, with the cell $(x,y')$ replaced by $v$. In the second case, when there is a cell West of $w$ in $C'$,
the argument is similar, with the NE-corner replaced by the NW-corner.

Case 2. $w$ is East of $v$.

 If $w$ is the East-most cell in this row, then  it must have a South neighbor, otherwise it would be a leaf.
Hence $v$ cannot have a South neighbor because then $w$ would be a NE-corner in a 4-cycle. However this contradicts the minimality of $C'$, as $w$ is an articulation cell. If $w$ is not the East-most cell in this row, then we consider the East-most cell in this row and the argument is as in Part 1, with the cell $(x,y')$ replaced by $v$.

Case 3. $w$ is South of $v$.

The cell $w$ cannot be a single node of $C'$ in this row because then $w$ would be a leaf of $C$. 
Hence either there is a cell East of $w$ or West of $w$ in $C'$. Consider the first case. Suppose that $w'\neq w$ is the East-most cell in this row.  If $w'$ is a neighbor of $w$ then it must have a South neighbor, otherwise it would be a leaf. Hence $w$ cannot have a South neighbor because then $w'$ would be a NE-corner in a 4-cycle. The rest of the argument is as in Part 1, with the cell $(x,y')$ replaced by $v$.
In the second case, when there is a cell West of $w$ in $C'$,
the argument is similar, with the NE-corner replaced by the NW-corner.

Case 4. $w$ is West of $v$.

The argument is analogous to Case 2.
\end{proof}

We are now ready to prove Theorem \ref{fa}. The proof consists in showing that Algorithm {\tt Connected Contractible Gathering} correctly performs gathering of any connected contractible configuration. In view of the formulation of procedures
{\tt Leaf Destruction} and {\tt 4-cycle Destruction}, this algorithm can be executed by copies of a deterministic finite automaton without initial input.

{\bf Proof of Theorem \ref{fa}}

In order to prove the correctness of Algorithm {\tt Connected Contractible Gathering}, consider any initial connected contractible configuration $C$.
We first prove that all agents eventually get to a single cell.
If the configuration has neither leaves nor 4-cycles containing NE-corners or NW-corners, then this is the case for the initial configuration $C$, by Lemma \ref{geo}.
Hence suppose that $C$ has either leaves or such corners. In every move an agent goes from a leaf to the neighbor cell, or from a corner as above to its South neighbor cell.
After a finite number of moves some cell occupied by agents becomes non-occupied. By the formulation of the algorithm, the opposite change never happens.
Moreover, whenever the number of occupied cells decreases, the configuration remains connected contractible. Indeed, destroying a leaf or a NE-corner or a NW-corner in a 4-cycle cannot disconnect a configuration.
To see that it cannot create a hole, notice that, if a configuration $C'$ results from configuration $C''$ by deleting a leaf or such a corner $x$, and if there were a hole $H$ in $C'$, then either $H$
or $H\setminus \{x\}$ must have been a hole in $C''$. It follows that whenever Algorithm {\tt Connected Contractible Gathering}  creates a new configuration, then it is connected contractible.
Eventually a connected contractible configuration is created that has neither leaves nor 4-cycles containing NE-corners or NW-corners. By Lemma \ref{geo}, all agents in this configuration are in a single cell.
The first look of any agent after this time confirms that there are no neighboring agents, and hence the agent transits to state $\omega$. This proves that, starting from any initial connected contractible configuration, all agents eventually get to a single cell and transit to state $\omega$.
Since the configuration never gets disconnected during the algorithm execution, agents can never transit to state $\omega$ when not all agents are in a single cell.  This proves the correctness of Algorithm {\tt Connected Contractible Gathering}.
\qed

\section{Conclusion}

We showed that ``three qualities are better than two'' for asynchronous gathering of agents in the planar grid: while initial input, unbounded memory and randomization put together can produce a machine  that can be used to gather all possible initial configurations with probability 1, no machine that has some two of these qualities but not the third, can be used to achieve this goal.
It would be interesting to give an exact characterization of classes of configurations that can be gathered using copies of machines enjoying every two of these qualities but not the third
(i.e., resp. deterministic Turing Machines with initial input, random finite automata with initial input, and random Turing machines without initial input), and to give
an exact characterization of classes of configurations that can be gathered using copies of machines that have a single of these qualities (i.e., resp. deterministic finite automata with initial input, 
deterministic Turing Machines without initial input, and random finite automata without initial input).  

For deterministic machines without initial input we showed that Turing Machines can gather all initial connected configurations and that finite automata can gather all initial connected contractible configurations.
This leaves an interesting open question: Can we use our weakest machines, i.e., deterministic finite automata without initial input, to gather all connected configurations. Our conjecture is: no. Regardless of the status of this conjecture,
what class of configurations can be gathered when agents are copies of a deterministic finite automaton without initial input?

\section{Acknowledgements}
This research was supported in part by NSERC discovery grant 8136 -- 2013 
and by the Research Chair in Distributed Computing of
the Universit\'{e} du Qu\'{e}bec en Outaouais.

\bibliographystyle{plain}

\begin{thebibliography}{12}
\bibitem{AGM}
C. Agathangelou, C. Georgiou, M. Mavronicolas, A distributed algorithm for gathering many fat mobile robots in the plane,
Proc. 32nd Annual ACM Symposium on
Principles of Distributed Computing (PODC 2013), 250-259.

\bibitem{AP06}
N. Agmon and D. Peleg, 
Fault-tolerant gathering algorithms for autonomous mobile robots,  
SIAM J. Comput. 36 (2006), 56-82. 

\bibitem{alpern95a}
S. Alpern,
The rendezvous search problem,
SIAM J. on Control and Optimization 33 (1995), 673-683.


\bibitem{alpern02a}
S. Alpern,
Rendezvous search on labelled networks,
Naval Reaserch Logistics 49 (2002), 256-274.

\bibitem{alpern02b}
S. Alpern and S. Gal,
The theory of search games and rendezvous.
Int. Series in Operations research and Management Science,
Kluwer Academic Publisher, 2002.




\bibitem{anderson90}
E. Anderson and R. Weber,
The rendezvous problem on discrete locations,
Journal of Applied Probability 28 (1990), 839-851.

\bibitem{anderson98a}
E. Anderson and S. Fekete,
Asymmetric rendezvous on the plane,
Proc. 14th Annual ACM Symp. on Computational Geometry (1998), 365-373.

\bibitem{anderson98b}
E. Anderson and S. Fekete,
Two-dimensional rendezvous search,
Operations Research 49 (2001), 107-118.

\bibitem{BCGIL}
E. Bampas, J. Czyzowicz, L. Gasieniec, D. Ilcinkas, A. Labourel,
Almost optimal asynchronous rendezvous in infinite multidimensional grids. 
Proc. 24th International Symposium on Distributed Computing (DISC 2010),  297-311.


\bibitem{baston98}
V. Baston and S. Gal,
Rendezvous on the line when the players' initial distance is
given by an unknown probability distribution,
SIAM J. on Control and Opt. 36 (1998), 1880-1889.

\bibitem{baston01}
V. Baston and S. Gal,
Rendezvous search when marks are left at the starting
points,
Naval Reaserch Logistics 48 (2001), 722-731.


\bibitem{CFPS}
M. Cieliebak, P. Flocchini, G. Prencipe, N. Santoro, 
Distributed Computing by Mobile Robots: Gathering, SIAM J. Comput. 41 (2012), 829-879.

\bibitem{CP05}
R. Cohen and D. Peleg, 
Convergence properties of the gravitational algorithm in asynchronous robot 
systems, SIAM J. Comput. 34 (2005), 1516-1528. 

\bibitem{CP08}
R. Cohen and D. Peleg, 
Convergence of autonomous mobile robots with inaccurate sensors and movements, 
SIAM J. Comput. 38 (2008), 276-302. 

\bibitem{CKP}
J. Czyzowicz, A. Kosowski, A. Pelc, How to meet when you forget: Log-space rendezvous in arbitrary graphs, Distributed Computing 25 (2012), 165-178. 

\bibitem{CLP}
J. Czyzowicz, A. Labourel, A. Pelc, How to meet asynchronously (almost) everywhere, 
ACM Transactions on Algorithms 8 (2012), article 37. 

\bibitem{DGKKP}
G. De Marco, L. Gargano, E. Kranakis, D. Krizanc, A. Pelc, U. Vaccaro, 
Asynchronous deterministic rendezvous in graphs, 
Theoretical Computer Science 355 (2006), 315-326.

\bibitem{DFKP}
A. Dessmark, P. Fraigniaud, D. Kowalski, A. Pelc.
Deterministic rendezvous in graphs.
Algorithmica 46 (2006), 69-96.

\bibitem{DP}
Y. Dieudonn\'{e}, A. Pelc, Anonymous meeting in networks, Algorithmica 74 (2016), 908-946 . 

\bibitem{DPP}
Y. Dieudonn\'{e}, A. Pelc, D. Peleg, Gathering despite mischief,  ACM Transactions on Algorithms 11 (2014), article 1. 


\bibitem{ELSUW}
Y. Emek, T. Langner, D. Stolz, J. Uitto, R. Wattenhofer, How many ants does it take to find the food? Theor. Comput. Sci. 608 (2015), 255-267.


\bibitem{fpsw}
P. Flocchini, G. Prencipe, N. Santoro, P. Widmayer,
Gathering of asynchronous oblivious robots with limited visibility,
 Theor. Comput. Sci. 337 (2005), 147-168.

\bibitem{FP2}
P. Fraigniaud, A. Pelc, Delays induce an exponential memory gap for rendezvous in trees, ACM Transactions on Algorithms 9 (2013), article 17. 


\bibitem{gal99}
S. Gal,
Rendezvous search on the line,
Operations Research 47 (1999), 974-976.

\bibitem{GWB}
N. Gordon, I. A. Wagner, A. M. Bruckstein,
Gathering Multiple Robotic Agents with Limited Sensing Capabilities,
Proc. ANTS Workshop 2004, 142-153.

\bibitem{GP}
S. Guilbault, A. Pelc, Asynchronous rendezvous of anonymous agents in arbitrary graphs, Proc. 15th International Conference on Principles of Distributed Systems (OPODIS 2011), LNCS 7109, 162-173. 

\bibitem{israeli}
A. Israeli and M. Jalfon, Token management schemes and random walks yield
self stabilizing mutual exclusion, Proc. 9th Annual ACM Symposium on
Principles of Distributed Computing (PODC 1990), 119-131.

\bibitem{KLUW}
B. Keller, T.  Langner, J. Uitto, R.  Wattenhofer,
Overcoming obstacles with ants,
Proc. 19th International Conference on Principles of Distributed Systems (OPODIS 2015), 1-17.



\bibitem{KKPM08}
E. Kranakis, D. Krizanc, and P. Morin, 
Randomized Rendez-Vous with Limited Memory,
Proc. 8th Latin American Theoretical Informatics (LATIN 2008), Springer LNCS 4957, 605-616.

\bibitem{KKSS}
E. Kranakis, D. Krizanc, N. Santoro and C. Sawchuk, 
Mobile agent rendezvous in a ring, 
Proc. 23rd Int. Conference on Distributed Computing Systems
(ICDCS 2003), IEEE, 592-599.

\bibitem{lim96}
W. Lim and S. Alpern,
Minimax rendezvous on the line,
SIAM J. on Control and Optimization 34 (1996), 1650-1665.


\bibitem{Pe}
A. Pelc, Deterministic rendezvous in networks: A comprehensive survey, Networks 59 (2012), 331-347. 



\bibitem{TSZ07}
A. Ta-Shma and U. Zwick.
Deterministic rendezvous, treasure hunts and strongly universal exploration sequences.
Proc. 18th ACM-SIAM Symposium on Discrete Algorithms (SODA 2007), 599-608.

\bibitem{thomas92}
L. Thomas,
Finding your kids when they are lost,
Journal on Operational Res. Soc. 43 (1992), 637-639.




\bibitem{YY}
X. Yu and M. Yung, 
Agent rendezvous: a dynamic symmetry-breaking problem, 
Proc.  International Colloquium on Automata,
Languages, and Programming (ICALP 1996), Springer LNCS 1099, 610-621.


\end{thebibliography}


\end{document}